\documentclass[12pt,a4paper,american]{article}
\usepackage{amsthm}
\usepackage{amsmath}
\usepackage{amssymb}

\makeatletter
\newcommand{\lyxaddress}[1]{
\par {\raggedright #1
\vspace{1.4em}
\noindent\par}
}
\theoremstyle{plain}
\newtheorem{thm}{Theorem}
  \theoremstyle{definition}
  \newtheorem{defn}[thm]{Definition}
  \theoremstyle{remark}
  \newtheorem{rem}[thm]{Remark}
  \theoremstyle{plain}
  \newtheorem{prop}[thm]{Proposition}
  \theoremstyle{plain}
  \newtheorem{lem}[thm]{Lemma}
 \theoremstyle{definition}
  \newtheorem{example}[thm]{Example}
  \theoremstyle{remark}
  \newtheorem*{rem*}{Remark}

\usepackage{mathrsfs}

\newcommand{\diag}{\mathop\mathrm{diag}\nolimits}
\renewcommand{\span}{\mathop\mathrm{span}\nolimits}

\newcommand{\spec}{\mathop\mathrm{spec}\nolimits}

\newcommand{\hs}[1]{\hspace{#1}}

\makeatother

\usepackage{babel}

\begin{document}

\title{On the eigenvalue problem for a particular class of finite Jacobi
matrices}

\author{F.~\v{S}tampach$^{1}$, P.~\v{S}\v{t}ov\'\i\v{c}ek$^{2}$}

\date{{}}

\maketitle

\lyxaddress{Department of Mathematics, Faculty of Nuclear Science, Czech Technical
University in Prague, Trojanova13, 12000 Praha, Czech Republic}

\lyxaddress{$^{1}$stampfra@fjfi.cvut.cz}

\lyxaddress{$^{2}$stovicek@kmlinux.fjfi.cvut.cz}
\begin{abstract}
\noindent A function $\mathfrak{F}$ with simple and nice algebraic
properties is defined on a subset of the space of complex sequences.
Some special functions are expressible in terms of $\mathfrak{F}$,
first of all the Bessel functions of first kind. A compact formula
in terms of the function $\mathfrak{F}$ is given for the determinant
of a Jacobi matrix. Further we focus on the particular class of Jacobi
matrices of odd dimension whose parallels to the diagonal are constant
and whose diagonal depends linearly on the index. A formula is derived
for the characteristic function. Yet another formula is presented
in which the characteristic function is expressed in terms of the
function $\mathfrak{F}$ in a simple and compact manner. A special
basis is constructed in which the Jacobi matrix becomes a sum of a
diagonal matrix and a rank-one matrix operator. A vector-valued function
on the complex plain is constructed having the property that its values
on spectral points of the Jacobi matrix are equal to corresponding
eigenvectors.
\end{abstract}
\vskip\baselineskip\noindent \emph{Keywords}: tridiagonal matrix,
finite Jacobi matrix, eigenvalue problem, characteristic function

\section{Introduction}

The results of the current paper are related to the eigenvalue problem
for finite-dimensional symmetric tridiagonal (Jacobi) matrices. Notably,
the eigenvalue problem for finite Jacobi matrices is solvable explicitly
in terms of generalized hypergeometric series \cite{KuznetsovSklyanin}.
Here we focus on a very particular class of Jacobi matrices which
makes it possible to derive some expressions in a comparatively simple
and compact form. We do not aim at all, however, at a complete solution
of the eigenvalue problem. We restrict ourselves to derivation of
several explicit formulas, first of all that for the characteristic
function, as explained in more detail below. We also develop some
auxiliary notions which may be, to our opinion, of independent interest.

First, we introduce a function, called $\mathfrak{F}$, defined on
a subset of the space of complex sequences. In the remainder of the
paper it is intensively used in various formulas. The function $\mathfrak{F}$
has remarkably simple and nice algebraic properties. Among others,
with the aid of $\mathfrak{F}$ one can relate an infinite continued
fraction to any sequence from the definition domain on which $\mathfrak{F}$
takes a nonzero value. This may be compared to the fact that there
exists a correspondence between infinite Jacobi matrices and infinite
continued fractions, as explained in \cite[Chp.~1]{Akhiezer}. Let
us also note that some special functions are expressible in terms
of $\mathfrak{F}$. First of all this concerns the Bessel functions
of first kind. We examine the relationship between $\mathfrak{F}$
and the Bessel functions and provide some supplementary details on
it.

Further we introduce an infinite antisymmetric matrix, with entries
indexed by integers, such that its every row or column obeys a second-order
difference equation which is very well known from the theory of Bessel
functions. With the aid of function $\mathfrak{F}$ one derives a
general formula for entries of this matrix. The matrix also plays
an essential role in the remainder of the paper.

As an application we present a comparatively simple formula for the
determinant of a Jacobi matrix of odd dimension under the assumption
that the neighboring parallels to the diagonal are constant. As far
as the determinant is concerned this condition is not very restrictive
since a Jacobi matrix can be written as a product of another Jacobi
matrix with all units on the neighboring parallels which is sandwiched
with two diagonal matrices. The formula further simplifies in the
particular case when the diagonal is antisymmetric (with respect to
its center). In that case zero is always an eigenvalue and we give
an explicit formula for the corresponding eigenvector.

Finally we focus on the rather particular class of Jacobi matrices
of odd dimension whose parallels to the diagonal are constant and
whose diagonal depends linearly on the index. Within this class it
suffices to consider matrices whose diagonal is, in addition, antisymmetric.
In this case we derive a formula for the characteristic function.
Yet another formula is presented in which the characteristic function
is expressed in terms of the function $\mathfrak{F}$ in a very simple
and compact manner. Moreover, we construct a basis in which the Jacobi
matrix becomes a sum of a diagonal matrix and a rank-one matrix operator.
This form is rather suitable for various computations. Particularly,
one can readily derive a formula for the resolvent. In addition, a
vector-valued function on the complex plain is constructed having
the property that its values on spectral points of the Jacobi matrix
are equal to corresponding eigenvectors.

\section{The function $\mathfrak{F}$}

We introduce a function $\mathfrak{F}$ defined on a subset of the
linear space formed by all complex sequences $x=\{x_{k}\}_{k=1}^{\infty}$.
\begin{defn}
Define $\mathfrak{F}:D\rightarrow\mathbb{C}$, \begin{equation}
\mathfrak{F}(x)=1+\sum_{m=1}^{\infty}(-1)^{m}\sum_{k_{1}=1}^{\infty}\,\sum_{k_{2}=k_{1}+2}^{\infty}\dots\,\sum_{k_{m}=k_{m-1}+2}^{\infty}\, x_{k_{1}}x_{k_{1}+1}x_{k_{2}}x_{k_{2}+1}\dots x_{k_{m}}x_{k_{m}+1}\label{eq:defF}\end{equation}
where \[
D=\left\{ \{x_{k}\}_{k=1}^{\infty};\,\sum_{k=1}^{\infty}|x_{k}x_{k+1}|<\infty\right\} .\]
For a finite number of complex variables we identify $\mathfrak{F}(x_{1},x_{2},\dots,x_{n})$
with $\mathfrak{F}(x)$ where $x=(x_{1},x_{2},\dots,x_{n},0,0,0,\dots)$.
By convention, we also put $\mathfrak{F}(\emptyset)=1$ where $\emptyset$
is the empty sequence.\end{defn}
\begin{rem}
Note that the domain $D$ is not a linear space. One has, however,
$\ell^{2}(\mathbb{N})\subset D$. To see that the series on the RHS
of (\ref{eq:defF}) converges absolutely whenever $x\in D$ observe
that the absolute value of the $m$th summand is majorized by the
expression\[
\sum_{\substack{\noalign{\smallskip}k\in\mathbb{N}^{m}\\
\noalign{\smallskip}k_{1}<k_{2}<\ldots<k_{m}}
}\,|x_{k_{1}}x_{k_{1}+1}x_{k_{2}}x_{k_{2}+1}\dots x_{k_{m}}x_{k_{m}+1}|\leq\frac{1}{m!}\!\left(\sum_{j=1}^{\infty}|x_{j}x_{j+1}|\right)^{\! m}.\]

\end{rem}
Obviously, if all but finitely many elements of a sequence $x$ are
zeroes then $\mathfrak{F}(x)$ reduces to a finite sum. Thus\begin{eqnarray*}
 &  & \mathfrak{F}(x_{1})=1,\mbox{ }\mathfrak{F}(x_{1},x_{2})=1-x_{1}x_{2},\mbox{ }\mathfrak{F}(x_{1},x_{2},x_{3})=1-x_{1}x_{2}-x_{2}x_{3},\\
 &  & \mathfrak{F}(x_{1},x_{2},x_{3},x_{4})=1-x_{1}x_{2}-x_{2}x_{3}-x_{3}x_{4}+x_{1}x_{2}x_{3}x_{4},\mbox{ etc.}\end{eqnarray*}

Let $T$ denote the truncation operator from the left defined on the
space of all sequences: \[
T(\{x_{k}\}_{k=1}^{\infty})=\{x_{k+1}\}_{k=1}^{\infty}.\]
$T^{n}$, $n=0,1,2,\dots$, stands for a power of $T$. Hence $T^{n}(\{x_{k}\}_{k=1}^{\infty})=\{x_{k+n}\}_{k=1}^{\infty}$.

The proof of the following proposition is immediate.
\begin{prop}
\label{thm:F_T_recur} For all $x\in D$ one has\begin{equation}
\mathfrak{F}(x)=\mathfrak{F}(Tx)-x_{1}x_{2}\,\mathfrak{F}(T^{2}x).\label{eq:F_T_recur}\end{equation}
Particularly, if $n\geq2$ then\begin{equation}
\mathfrak{F}(x_{1},x_{2},x_{3},\dots,x_{n})=\mathfrak{F}(x_{2},x_{3},\dots,x_{n})-x_{1}x_{2}\,\mathfrak{F}(x_{3},\dots,x_{n}).\label{eq:F_finite_recur}\end{equation}
\end{prop}
\begin{rem}
Clearly, given that $\mathfrak{F}(\emptyset)=\mathfrak{F}(x_{1})=1$,
relation (\ref{eq:F_finite_recur}) determines recursively and unambiguously
$\mathfrak{F}(x_{1},\ldots,x_{n})$ for any finite number of variables
$n\in\mathbb{Z}_{+}$ (including $n=0$).
\end{rem}
\begin{rem}
One readily verifies that\begin{equation}
\mathfrak{F}(x_{1},x_{2},\dots,x_{n})=\mathfrak{F}(x_{n},\dots,x_{2},x_{1}).\label{eq:Fx_reversed}\end{equation}
Hence equality (\ref{eq:F_finite_recur}) implies, again for $n\geq2$,\begin{equation}
\mathfrak{F}(x_{1},\dots,x_{n-2},x_{n-1},x_{n})=\mathfrak{F}(x_{1},\dots,x_{n-2},x_{n-1})-x_{n-1}x_{n}\,\mathfrak{F}(x_{1},\dots,x_{n-2}).\label{eq:F_inverse_recur}\end{equation}

\end{rem}
\begin{rem}
\label{rem:F_FT_con_frac} For a given $x\in D$ such that $\mathfrak{F}(x)\neq0$
let us introduce sequences $\{P_{k}\}_{k=0}^{\infty}$ and $\{Q_{k}\}_{k=0}^{\infty}$
by $P_{0}=0$ and $P_{k}=\mathfrak{F}(x_{2},\ldots,x_{k})$ for $k\geq1$,
$Q_{k}=\mathfrak{F}(x_{1},\ldots,x_{k})$ for $k\geq0$. According
to (\ref{eq:F_inverse_recur}), the both sequences obey the difference
equation\[
Y_{k+1}=Y_{k}-x_{k}x_{k+1}Y_{k-1},\mbox{ }k=1,2,3,\ldots,\]
with the initial conditions $P_{0}=0$, $P_{1}=1$, $Q_{0}=Q_{1}=1$,
and define the infinite continued fraction\[
\frac{\mathfrak{F}(Tx)}{\mathfrak{F}(x)}=\lim_{k\to\infty}\,\frac{P_{k}}{Q_{k}}=\dfrac{1}{1-\dfrac{x_{1}x_{2}\phantom{I}}{1-\dfrac{x_{2}x_{3}\phantom{I}}{1-\dfrac{x_{3}x_{4}\phantom{I}}{1-\ldots}}}}\,.\]

\end{rem}
Proposition~\ref{thm:F_T_recur} admits a generalization.
\begin{prop}
For every $x\in D$ and $k\in\mathbb{N}$ one has\begin{equation}
\mathfrak{F}(x)=\mathfrak{F}(x_{1},\ldots,x_{k})\,\mathfrak{F}(T^{k}x)-\mathfrak{F}(x_{1},\ldots,x_{k-1})x_{k}x_{k+1}\,\mathfrak{F}(T^{k+1}x).\label{eq:F_T_recur_k}\end{equation}
\end{prop}
\begin{proof}
Let us proceed by induction in $k$. For $k=1$, equality (\ref{eq:F_T_recur_k})
coincides with (\ref{eq:F_T_recur}). Suppose (\ref{eq:F_T_recur_k})
is true for $k\in\mathbb{N}$. Applying Proposition~\ref{thm:F_T_recur}
to the sequence $T^{k}x$ and using (\ref{eq:F_inverse_recur}) one
finds that the RHS of (\ref{eq:F_T_recur_k}) equals\begin{eqnarray*}
 &  & \mathfrak{F}(x_{1},\ldots,x_{k})\,\mathfrak{F}(T^{k+1}x)-\mathfrak{F}(x_{1},\ldots,x_{k})x_{k+1}x_{k+2}\,\mathfrak{F}(T^{k+2}x)\\
 &  & \phantom{=\,}-\,\mathfrak{F}(x_{1},\ldots,x_{k-1})x_{k}x_{k+1}\,\mathfrak{F}(T^{k+1}x)\\
 &  & =\,\mathfrak{F}(x_{1},\ldots,x_{k},x_{k+1})\,\mathfrak{F}(T^{k+1}x)-\mathfrak{F}(x_{1},\ldots,x_{k})x_{k+1}x_{k+2}\,\mathfrak{F}(T^{k+2}x).\end{eqnarray*}
This concludes the verification.\end{proof}
\begin{rem}
With the aid of Proposition~\ref{thm:F_T_recur} one can rewrite
equality (\ref{eq:F_T_recur_k}) as follows\begin{equation}
\mathfrak{F}(x)=\mathfrak{F}(x_{1},\ldots,x_{k})\,\mathfrak{F}\!\left(\frac{\mathfrak{F}(x_{1},\ldots,x_{k-1})}{\mathfrak{F}(x_{1},\ldots,x_{k})}\, x_{k},x_{k+1},x_{k+2},x_{k+3},\ldots\right).\label{eq:Fx_eq_Fxk_Fxrem}\end{equation}

\end{rem}
Later on, we shall also need the following identity.
\begin{lem}
For any $n\in\mathbb{N}$ one has\begin{eqnarray}
 &  & u_{1}\mathfrak{F}(u_{2},u_{3},\ldots,u_{n})\mathfrak{F}(v_{1},v_{2},v_{3},\ldots,v_{n})-v_{1}\mathfrak{F}(u_{1},u_{2},u_{3},\ldots,u_{n})\mathfrak{F}(v_{2},v_{3},\ldots,v_{n})\nonumber \\
 &  & =\,\sum_{j=1}^{n}\left(\prod_{k=1}^{j-1}u_{k}v_{k}\right)\!(u_{j}-v_{j})\,\mathfrak{F}(u_{j+1},u_{j+2},\ldots,u_{n})\mathfrak{F}(v_{j+1},v_{j+2},\ldots,v_{n}).\label{eq:uFuFv_antisym}\end{eqnarray}
\end{lem}
\begin{proof}
The equality can be readily proved by induction in $n$ with the aid
of (\ref{eq:F_finite_recur}).\end{proof}
\begin{example}
For $t,w\in\mathbb{C}$, $|t|<1$, a simple computation leads to the
equality\begin{equation}
\mathfrak{F}\!\left(\left\{ t^{k-1}w\right\} _{k=1}^{\infty}\right)=1+\sum_{m=1}^{\infty}(-1)^{m}\,\frac{t^{m(2m-1)}w^{2m}}{(1-t^{2})(1-t^{4})\cdots(1-t^{2m})}\,.\label{eq:F_geom_tw}\end{equation}
This function can be identified with a basic hypergeometric series
(also called q-hypergeometric series) defined by\[
_{r}\phi_{s}(a;b;q,z)=\sum_{k=0}^{\infty}\frac{\left(a_{1};q\right)_{k}\ldots\left(a_{r};q\right)_{k}}{\left(b_{1};q\right)_{k}\ldots\left(b_{s};q\right)_{k}}\left((-1)^{k}q^{\frac{1}{2}k(k-1)}\right)^{\!1+s-r}\frac{z^{k}}{(q;q)_{k}}\]
where $r,s\in\mathbb{Z}_{+}$ (nonnegative integers) and\[
(\alpha;q)_{k}=\prod_{j=0}^{k-1}\left(1-\alpha q^{j}\right),\mbox{ }k=0,1,2,\ldots,\]
see \cite{GasperRahman}. In fact, the RHS in (\ref{eq:F_geom_tw})
equals $_{0}\phi_{1}(;0;t^{2},-tw^{2})$ where\[
_{0}\phi_{1}(;0;q,z)=\sum_{k=0}^{\infty}\frac{q^{k(k-1)}}{(q;q)_{k}}\, z^{k}=\sum_{k=0}^{\infty}\frac{q^{k(k-1)}}{(1-q)(1-q^{2})\ldots(1-q^{k})}\, z^{k},\]
with $q,z\in\mathbb{C}$, $|q|<1$, and the recursive rule (\ref{eq:F_T_recur})
takes the form\begin{equation}
_{0}\phi_{1}(;0;q,z)={}_{0}\phi_{1}(;0;q,qz)+z\,{}_{0}\phi_{1}(;0;q,q^{2}z).\label{eq:0phi1_recur}\end{equation}
Put $e(q;z)={}_{0}\phi_{1}(;0;q,(1-q)z)$. Then $\lim_{q\uparrow1}e(q;z)=\exp(z)$.
Hence $e(q;z)$ can be regarded as a q-deformed exponential function
though this is not the standard choice (compare with \cite{GasperRahman}
or \cite{KlimykSchmudgen} and references therein). Equality (\ref{eq:0phi1_recur})
can be interpreted as the discrete derivative\[
\frac{e(q;z)-e(q;qz)}{(1-q)z}=e(q;q^{2}z).\]
Moreover, in view of Remark~\ref{rem:F_FT_con_frac}, one has\[
\dfrac{1}{1+\dfrac{z}{1+\dfrac{qz}{1+\dfrac{q^{2}z}{1+\ldots}}}}=\frac{{}_{0}\phi_{1}(;0;q,qz)}{{}_{0}\phi_{1}(;0;q,z)}\,.\]
This equality is related to the Rogers-Ramanujan identities, see the
discussion in \cite[Chp.~7]{Andrews}.
\end{example}
\begin{example}
\label{rem:BesselJ_rel_F} The Bessel functions of the first kind
can be expressed in terms of function $\mathfrak{F}$. More precisely,
for $\nu\notin-\mathbb{N}$, one has\begin{equation}
J_{\nu}(2w)=\frac{w^{\nu}}{\Gamma(\nu+1)}\,\mathfrak{F}\!\left(\left\{ \frac{w}{\nu+k}\right\} _{k=1}^{\infty}\right).\label{eq:BesselJ_rel_F}\end{equation}
The recurrence relation (\ref{eq:F_T_recur}) transforms to the well
known identity\[
zJ_{\nu}(z)-2(\nu+1)J_{\nu+1}(z)+zJ_{\nu+2}(z)=0.\]
To prove (\ref{eq:BesselJ_rel_F}) one can proceed by induction in
$j=0,1,\ldots,m-1$, to show that\begin{eqnarray*}
 &  & \sum_{k_{1}=1}^{\infty}\,\sum_{k_{2}=k_{1}+2}^{\infty}\dots\,\sum_{k_{m}=k_{m-1}+2}^{\infty}\\
 &  & \quad\times\,\frac{1}{(\nu+k_{1})(\nu+k_{1}+1)(\nu+k_{2})(\nu+k_{2}+1)\ldots(\nu+k_{m})(\nu+k_{m}+1)}\\
 &  & =\,\frac{1}{j!}\,\sum_{k_{1}=1}^{\infty}\,\sum_{k_{2}=k_{1}+2}^{\infty}\dots\,\sum_{k_{m-j}=k_{m-j-1}+2}^{\infty}\\
 &  & \quad\times\,\frac{1}{(\nu+k_{1})(\nu+k_{1}+1)(\nu+k_{2})(\nu+k_{2}+1)\ldots(\nu+k_{m-j})(\nu+k_{m-j}+1)}\\
 &  & \quad\times\,\frac{1}{(\nu+k_{m-j}+2)(\nu+k_{m-j}+3)\ldots(\nu+k_{m-j}+j+1)}\,.\end{eqnarray*}
In particular, for $j=m-1$, the RHS equals\begin{eqnarray*}
 &  & \frac{1}{(m-1)!}\,\sum_{k_{1}=1}^{\infty}\,\frac{1}{(\nu+k_{1})(\nu+k_{1}+1)(\nu+k_{1}+2)\ldots(\nu+k_{1}+m)}\\
 &  & =\,\frac{1}{m!\,(\nu+1)(\nu+2)\ldots(\nu+m)}=\frac{\Gamma(\nu+1)}{m!\,\Gamma(\nu+m+1)}\end{eqnarray*}
and so\[
\frac{w^{\nu}}{\Gamma(\nu+1)}\,\mathfrak{F}\!\left(\left\{ \frac{w}{\nu+k}\right\} _{k=1}^{\infty}\right)=\sum_{m=0}^{\infty}(-1)^{m}\,\frac{w^{2m+\nu}}{m!\,\Gamma(\nu+m+1)}\,,\]
as claimed. Furthermore, Remark~\ref{rem:F_FT_con_frac} provides
us with the infinite fraction\[
\frac{\nu+1}{w}\,\frac{J_{\nu+1}(2w)}{J_{\nu}(2w)}=\dfrac{1}{1-\dfrac{\dfrac{w^{2}}{(\nu+1)(\nu+2)}}{1-\dfrac{\dfrac{w^{2}}{(\nu+2)(\nu+3)}}{1-\dfrac{\dfrac{w^{2}}{(\nu+3)(\nu+4)}}{1-\dots}}}}\,.\]
This can be rewritten as\[
\frac{J_{\nu+1}(z)}{J_{\nu}(z)}=\dfrac{z}{2(\nu+1)-\dfrac{z^{2}}{2(\nu+2)-\dfrac{z^{2}}{2(\nu+3)-\dfrac{z^{2}}{2(\nu+4)-\ldots}}}}\,.\]

\end{example}
Comparing to Example~\ref{rem:BesselJ_rel_F}, one can also find
the value of $\mathbb{\mathfrak{F}}$ on the truncated sequence $\{w/(\nu+k)\}_{k=1}^{n}$.
\begin{prop}
\label{thm:Fw_over_n_trunc} For $n\in\mathbb{Z}_{+}$ and $\nu\in\mathbb{C}\setminus\{-n,-n+1,\ldots,-1\}$
one has\begin{equation}
\mathfrak{F}\!\left(\frac{w}{\nu+1},\frac{w}{\nu+2},\dots,\frac{w}{\nu+n}\right)\!=\frac{\Gamma(\nu+1)}{\Gamma(\nu+n+1)}\sum_{s=0}^{[n/2]}\,(-1)^{s}\,\frac{(n-s)!}{s!\,(n-2s)!}\, w^{2s}\,\prod_{j=s}^{n-1-s}(\nu+n-j).\label{eq:Fw_over_nu_trunc}\end{equation}
In particular, for $m,n\in\mathbb{Z}_{+}$, $m\leq n$, one has \begin{equation}
\mathfrak{F}\!\left(\frac{w}{m+1},\frac{w}{m+2},\dots,\frac{w}{n}\right)\!=\frac{m!}{n!}\,\sum_{s=0}^{[(n-m)/2]}\,(-1)^{s}\,\frac{(n-s)!\,(n-m-s)!}{s!\,(m+s)!\,(n-m-2s)!}\, w^{2s}.\label{eq:Fw_over_n_trunc}\end{equation}
\end{prop}
\begin{proof}
Firstly, the equality\begin{eqnarray}
 &  & \sum_{k=1}^{n}\,\frac{(n+1-k)(n+2-k)\dots(n+s-1-k)}{(\nu+k)(\nu+k+1)\dots(\nu+k+s)}\nonumber \\
 &  & =\frac{n\,(n+1)\dots(n+s-1)}{s\,(\nu+n+s)\,(\nu+1)(\nu+2)\dots(\nu+s)}\label{eq:Fw_over_n_trunc_aux1}\end{eqnarray}
holds for all $n\in\mathbb{Z}_{+}$, $\nu\in\mathbb{C}$, $\nu\notin-\mathbb{N}$,
and $s\in\mathbb{N}$. To show (\ref{eq:Fw_over_n_trunc_aux1}) one
can proceed by induction in $s$. The case $s=1$ is easy to verify.
For the induction step from $s-1$ to $s$, with $s>1$, let us denote
the LHS of (\ref{eq:Fw_over_n_trunc_aux1}) by $Y_{s}(\nu,n)$. One
observes that \[
Y_{s}(\nu,n)=\frac{\nu+n+s-1}{s}\, Y_{s-1}(\nu,n)-\frac{\nu+n+2s-1}{s}\, Y_{s-1}(\nu+1,n).\]
Applying the induction hypothesis the equality readily follows.

Next one shows that\begin{eqnarray}
 &  & \hs{-2em}\sum_{k_{1}=1}^{n-2s+2}\,\sum_{k_{2}=k_{1}+2}^{n-2s+4}\,\dots\,\sum_{k_{s}=k_{s-1}+2}^{n}\nonumber \\
 &  & \times\,\frac{1}{(\nu+k_{1})(\nu+k_{1}+1)(\nu+k_{2})(\nu+k_{2}+1)\dots(\nu+k_{s})(\nu+k_{s}+1)}\label{eq:Fw_over_n_trunc_aux2}\\
 &  & \hs{-2em}=\,\frac{(n-2s+2)(n-2s+3)\dots(n-s+1)}{s!\,(\nu+1)(\nu+2)\dots(\nu+s)\,(\nu+n-s+2)(\nu+n-s+3)\dots(\nu+n+1)}\nonumber \end{eqnarray}
holds for all $n\in\mathbb{Z}_{+}$, $s\in\mathbb{N}$, $2s\leq n+2$.
To this end, we again proceed by induction in $s$. The case $s=1$
is easy to verify. In the induction step from $s-1$ to $s$, with
$s>1$, one applies the induction hypothesis to the LHS of (\ref{eq:Fw_over_n_trunc_aux2})
and arrives at the expression\begin{eqnarray*}
 &  & \hs{-1.5em}\sum_{k=1}^{n-2s+2}\,\frac{1}{(\nu+k)(\nu+k+1)\,(s-1)!}\\
 &  & \hs{-1em}\times\,\frac{(n-k-2s+3)(n-k-2s+4)\dots(n-k-s+1)}{(\nu+k+2)(\nu+k+3)\dots(\nu+k+s)\,(\nu+n-s+3)(\nu+n-s+4)\dots(\nu+n+1)}.\end{eqnarray*}
Using (\ref{eq:Fw_over_n_trunc_aux1}) one obtains the RHS of (\ref{eq:Fw_over_n_trunc_aux2}),
as claimed.

Finally, to conclude the proof, it suffices to notice that\begin{eqnarray*}
 &  & \hs{-1.5em}\mathfrak{F}\!\left(\frac{w}{\nu+1},\frac{w}{\nu+2},\dots,\frac{w}{\nu+n}\right)=\,1+\sum_{s=1}^{[n/2]}\,(-1)^{s}\,\sum_{k_{1}=1}^{n-2s+1}\,\sum_{k_{2}=k_{1}+2}^{n-2s+3}\dots\sum_{k_{s}=k_{s-1}+2}^{n-1}\\
 &  & \hs{6em}\times\,\frac{w^{2s}}{(\nu+k_{1})(\nu+k_{1}+1)(\nu+k_{2})(\nu+k_{2}+1)\dots(\nu+k_{s})(\nu+k_{s}+1)}\end{eqnarray*}
and to use equality (\ref{eq:Fw_over_n_trunc_aux2}).
\end{proof}
One can complete Proposition~\ref{thm:Fw_over_n_trunc} with another
relation to Bessel functions.
\begin{prop}
\label{thm:BesselJY_rel_F} For $m,n\in\mathbb{Z}_{+}$, $m\leq n$,
one has\begin{eqnarray}
 &  & \pi J_{m}(2w)Y_{n+1}(2w)\,=\,-\frac{n!}{m!}\, w^{m-n-1}\,\mathfrak{F}\!\left(\frac{w}{m+1},\frac{w}{m+2},\dots,\frac{w}{n}\right)\label{eq:BesselJY_rel_F}\\
 &  & \qquad-\,\sum_{s=0}^{m-1}\frac{(m-s-1)!\,(n-m+2s+1)!}{s!\,(n+s+1)!\,(n-m+s+1)!}\, w^{n-m+2s+1}+O\big(w^{m+n+1}\log(w)\big).\nonumber \end{eqnarray}
\end{prop}
\begin{proof}
Recall the following two facts from the theory of Bessel functions
(see, for instance, \cite[Chapter~VII]{Erdelyi_etal}). Firstly, for
$\mu,\nu\notin-\mathbb{N}$, one has \[
J_{\mu}(z)J_{\nu}(z)=\sum_{s=0}^{\infty}(-1)^{s}\,\frac{(s+\mu+\nu+1)_{s}}{s!\,\Gamma(\mu+s+1)\Gamma(\nu+s+1)}\left(\frac{z}{2}\right)^{\mu+\nu+2s}\]
where $(a)_{s}=a(a+1)\ldots(a+s-1)$ is the Pochhammer symbol. Secondly,
for $n\in\mathbb{Z}_{+}$,\[
\pi Y_{n}(z)=\frac{\partial}{\partial\nu}\left(J_{\nu}(z)-(-1)^{n}J_{-\nu}(z)\right)\Bigg|_{\nu=n}.\]
For $m,n\in\mathbb{Z}_{+}$, $m\leq n$, a straightforward computation
based on these facts yields\begin{eqnarray}
 &  & \hs{-2.5em}\pi J_{m}(z)Y_{n}(z)\,=\,-\,\sum_{s=0}^{[(n-m-1)/2]}\,(-1)^{s}\,\frac{(n-s-1)!\,(n-m-s-1)!}{s!\,(m+s)!\,(n-m-2s-1)!}\left(\frac{z}{2}\right)^{m-n+2s}\nonumber \\
 &  & \hs{-0.5em}-\,\sum_{s=0}^{m-1}\frac{(m-s-1)!\,(n-m+2s)!}{s!\,(n+s)!\,(n-m+s)!}\left(\frac{z}{2}\right)^{n-m+2s}\,+\,2J_{m}(z)J_{n}(z)\log\!\left(\frac{z}{2}\right)\label{eq:BesselJY_series}\\
 &  & \hs{-0.5em}+\,\sum_{s=0}^{\infty}(-1)^{s}\,\frac{(m+n+2s)!}{s!\,(m+s)!\,(n+s)!\,(m+n+s)!}\left(\frac{z}{2}\right)^{m+n+2s}\Big(2\psi(m+n+2s+1)\nonumber \\
 &  & \hs{5.4em}-\psi(m+s+1)-\psi(n+s+1)-\psi(m+n+s+1)-\psi(s+1)\Big)\nonumber \end{eqnarray}
where $\psi(z)=\Gamma'(z)/\Gamma(z)$ is the digamma function. The
proposition follows from (\ref{eq:BesselJY_series}) and (\ref{eq:Fw_over_n_trunc}).\end{proof}
\begin{rem}
\label{rem:BesselJY_n_eq_m-1} Note that the first term on the RHS
of (\ref{eq:BesselJY_rel_F}) contains only negative powers of $w$.
One can extend (\ref{eq:BesselJY_rel_F}) to the case $n=m-1$. Then\[
\pi J_{m}(2w)Y_{m}(2w)=-\,\sum_{s=0}^{m-1}\frac{(m-s-1)!\,(2s)!}{(s!)^{2}\,(m+s)!}\, w^{2s}+O\big(w^{2m}\log(w)\big).\]

\end{rem}

\section{The matrix $\mathfrak{J}$}

In this section we introduce an infinite matrix $\mathfrak{J}$ that
is basically determined by two simple properties -- it is antisymmetric
and its every row satisfies a second-order difference equation known
from the theory of Bessel functions. Of course, in that case every
column of the matrix satisfies the difference equation as well.
\begin{lem}
Suppose $w\in\mathbb{C}\setminus\{0\}$. The dimension of the vector
space formed by infinite-dimensional matrices $A=\{A(m,n)\}_{m,n\in\mathbb{Z}}$
satisfying, for all $m,n\in\mathbb{Z}$,\begin{equation}
wA(m,n-1)-nA(m,n)+wA(m,n+1)=0\label{eq:A_order2_diff_eq}\end{equation}
and\begin{equation}
A(n,m)=-A(m,n),\label{eq:A_antisymm}\end{equation}
equals $1$. Every such a matrix is unambiguously determined by the
value $A(0,1)$, and one has\begin{equation}
\forall n\in\mathbb{Z},\mbox{ }A(n,n+1)=A(0,1).\label{eq:A_pardiag_const}\end{equation}
\end{lem}
\begin{proof}
Suppose $A$ solves (\ref{eq:A_order2_diff_eq}) and (\ref{eq:A_antisymm}).
Then $A(m,m)=0$. Equating $m=n$ in (\ref{eq:A_order2_diff_eq})
and using (\ref{eq:A_antisymm}) one finds that $A(n,n+1)=-A(n,n-1)=A(n-1,n)$.
Hence (\ref{eq:A_pardiag_const}) is fulfilled. Clearly, the matrix
$A$ is unambiguously determined by the second-order difference equation
(\ref{eq:A_order2_diff_eq}) in $n$ and by the initial conditions
$A(m,m)=0$, $A(m,m+1)=A(0,1)$, when $m$ runs through $\mathbb{Z}$.

Conversely, choose $\lambda\in\mathbb{C}$, $\lambda\neq0$. Let $A$
be the unique matrix determined by (\ref{eq:A_order2_diff_eq}) and
the initial conditions $A(m,m)=0$, $A(m,m+1)=\lambda$. It suffices
to show that $A$ satisfies (\ref{eq:A_antisymm}) as well. Note that
$A(m,m-1)=-\lambda$. Furthermore,\begin{eqnarray*}
 &  & wA(m-1,m+1)-mA(m,m+1)+wA(m+1,m+1)\\
 &  & =\, wA(m-1,m+1)-mA(m-1,m)+wA(m-1,m-1)\\
 &  & =\,0.\end{eqnarray*}
From (\ref{eq:A_order2_diff_eq}) and the initial conditions it follows
that $A(m,m+2)=(m+1)\lambda/w$, and so $mA(m,m+2)=(m+1)A(m-1,m+1)$.
Consequently,\begin{eqnarray*}
 &  & wA(m-1,m+2)-mA(m,m+2)+wA(m+1,m+2)\\
 &  & =\, wA(m-1,m+2)-(m+1)A(m-1,m+1)+wA(m-1,m)\\
 &  & =\,0.\end{eqnarray*}
One observes that, for a given $m\in\mathbb{Z}$, the sequence\[
x_{n}=-A(m-1,n)+\frac{m}{w}\, A(m,n),\mbox{ }n\in\mathbb{Z},\]
solves the difference equation \begin{equation}
wx_{n-1}-nx_{n}+wx_{n+1}=0\label{eq:diff_eq_x}\end{equation}
 with the initial conditions $x_{m+1}=A(m+1,m+1)$, $x_{m+2}=A(m+1,m+2)$.
By the uniqueness, $x_{n}=A(m+1,n)$. This means that, for all $m,n\in\mathbb{Z}$,\[
wA(m-1,n)-mA(m,n)+wA(m+1,n)=0.\]
Put $B(m,n)=-A(n,m)$. Then $B$ fulfills (\ref{eq:A_order2_diff_eq})
and $B(m,m)=0$, $B(m,m+1)=\lambda$. Whence $B=A$.\end{proof}
\begin{lem}
\label{thm:A_odd_each_var} Suppose $w\in\mathbb{C}\setminus\{0\}$.
If a matrix $A=\{A(m,n)\}_{m,n\in\mathbb{Z}}$ satisfies (\ref{eq:A_order2_diff_eq})
and (\ref{eq:A_antisymm}) then\begin{equation}
\forall m,n\in\mathbb{Z},\mbox{ }A(m,-n)=(-1)^{n}A(m,n),\mbox{ }A(-m,n)=(-1)^{m}A(m,n).\label{eq:A_odd_each_var}\end{equation}
\end{lem}
\begin{proof}
For any sequence $\{x_{n}\}_{n\in\mathbb{Z}}$ satisfying the difference
equation (\ref{eq:diff_eq_x}) one can verify, by mathematical induction,
that $x_{-n}=(-1)^{n}x_{n}$, $n=0,1,2,\ldots$.\end{proof}
\begin{defn}
For a given parameter $w\in\mathbb{C}\setminus\{0\}$ let $\mathfrak{J}=\{\mathfrak{J}(m,n)\}_{m,n\in\mathbb{Z}}$
denote the unique matrix satisfying (\ref{eq:A_order2_diff_eq}),
(\ref{eq:A_antisymm}) and $\mathfrak{J}(m,m+1)=1$, $\forall m\in\mathbb{Z}$.\end{defn}
\begin{rem}
Here are several particular entries of the matrix $\mathfrak{J}$,\[
\mathfrak{J}(m,m)=0,\,\mathfrak{J}(m,m+1)=1,\,\mathfrak{J}(m,m+2)=\frac{m+1}{w},\,\mathfrak{J}(m,m+3)=\frac{(m+1)(m+2)}{w^{2}}-1,\]
with $m\in\mathbb{Z}$. Some other particular values follow from (\ref{eq:A_antisymm})
and (\ref{eq:A_odd_each_var}). Below, in Proposition~\ref{thm:Jmn},
we derive a general formula for $\mathfrak{J}(m,n)$.\end{rem}
\begin{lem}
For $0\leq m<n$ one has (with the convention $\mathfrak{F}(\emptyset)$=1)\begin{equation}
\mathfrak{J}(m,n)=\frac{(n-1)!}{m!}\, w^{m-n+1}\,\mathfrak{F}\!\left(\frac{w}{m+1},\frac{w}{m+2},\ldots,\frac{w}{n-1}\right).\label{eq:J_rel_F}\end{equation}
\end{lem}
\begin{proof}
The RHS of (\ref{eq:J_rel_F}) equals 1 for $n=m+1$, and $(m+1)/w$
for $n=m+2$. Moreover, in view of (\ref{eq:F_inverse_recur}), the
RHS satisfies the difference equation (\ref{eq:diff_eq_x}) in the
index $n$.\end{proof}
\begin{rem}
From (\ref{eq:J_rel_F}) and (\ref{eq:BesselJ_rel_F}) it follows
that\[
\forall m\in\mathbb{Z},\mbox{ }\lim_{n\to\infty}\frac{w^{n-1}}{(n-1)!}\,\mathfrak{J}(m,n)=J_{m}(2w).\]
This is in agreement with the well known fact that, for any $w\in\mathbb{C}$,
the sequence $\{J_{n}(2w)\}_{n\in\mathbb{Z}}$ fulfills the second-order
difference equation (\ref{eq:diff_eq_x}).
\end{rem}
\begin{rem}
Rephrasing Proposition~\ref{thm:BesselJY_rel_F} and Remark~\ref{rem:BesselJY_n_eq_m-1}
one has, for $m,n\in\mathbb{Z}_{+}$, $m\leq n$,\begin{eqnarray*}
\pi J_{m}(2w)Y_{n}(2w) & = & -w^{-1}\mathfrak{J}(m,n)-\,\sum_{s=0}^{m-1}\frac{(m-s-1)!\,(n-m+2s)!}{s!\,(n+s)!\,(n-m+s)!}\, w^{n-m+2s}\\
 &  & +\, O\big(w^{m+n}\log(w)\big).\end{eqnarray*}

\end{rem}
Since, by definition, the matrix $\mathfrak{J}$ is antisymmetric
it suffices to determine the values $\mathfrak{J}(m,n)$ for $m\leq n$,
$m,n\in\mathbb{Z}$. In the derivations to follow as well as in the
remainder of the paper we use the Newton symbol in the usual sense,
i.e. for any $z\in\mathbb{C}$ and $n\in\mathbb{Z}_{+}$ we put\[
\binom{z}{n}=\frac{z(z-1)\ldots(z-n+1)}{n!}\,.\]

\begin{prop}
\label{thm:Jmn} For $m,n\in\mathbb{Z}$, $m\leq n$, one has\begin{equation}
\mathfrak{J}(m,n)=\sum_{s=0}^{\left[(n-m-1)/2\right]}\,(-1)^{s}\binom{n-s-1}{n-m-2s-1}\frac{(n-m-s-1)!}{s!}\, w^{m-n+2s+1}.\label{eq:Jmn}\end{equation}
\end{prop}
\begin{proof}
We distinguish several cases. First, consider the case $0\leq m<n$.
Then (\ref{eq:Jmn}) follows from (\ref{eq:J_rel_F}) and (\ref{eq:Fw_over_n_trunc}).
Observe also that for $m=n$, $m,n\in\mathbb{Z}$, the RHS of (\ref{eq:Jmn})
is an empty sum and so the both sides in (\ref{eq:Jmn}) are equal
to $0$.

Second, consider the case $m\leq0\leq n$. Put $m=-k$, $k\in\mathbb{Z}_{+}$.
The RHS of (\ref{eq:Jmn}) becomes\begin{equation}
\sum_{s=0}^{\left[(n+k-1)/2\right]}\,(-1)^{s}\binom{n-s-1}{n+k-2s-1}\frac{(n+k-s-1)!}{s!}\, w^{-k-n+2s+1}.\label{eq:Jmn_aux1}\end{equation}
Suppose $k\leq n$. Then the summands in (\ref{eq:Jmn_aux1}) vanish
for $s=0,1,\ldots,k-1$, and so the sum equals\[
\sum_{s=0}^{\left[(n-k-1)/2\right]}\,(-1)^{s+k}\frac{(n-k-s-1)!}{(n-k-2s-1)!\, s!}\,\frac{(n-s-1)!}{(s+k)!}\, w^{k-n+2s+1}.\]
By the first step, this expression is equal to $(-1)^{k}\mathfrak{J}(k,n)=\mathfrak{J}(-k,n)$
(see Lemma~\ref{thm:A_odd_each_var}). Further, suppose $k\geq n$.
Then the summands in (\ref{eq:Jmn_aux1}) vanish for $s=0,1,\ldots,n-1$,
and so the sum equals\[
\sum_{s=0}^{\left[(k-n-1)/2\right]}\,(-1)^{n+s}\binom{-s-1}{k-n-2s-1}\frac{(k-s-1)!}{(n+s)!}\, w^{n-k+2s+1}.\]
Using once more the first step, this expression is readily seen to
be equal to\linebreak $(-1)^{k+1}\mathfrak{J}(n,k)=\mathfrak{J}(-k,n)$.

Finally, consider the case $m\leq n\leq0$. Put $m=-k$, $n=-\ell$,
$k,\ell\in\mathbb{Z}_{+}$. Hence $0\leq\ell\leq k$. The RHS of (\ref{eq:Jmn})
becomes\[
\sum_{s=0}^{\left[(k-\ell-1)/2\right]}\,(-1)^{s}\binom{-\ell-s-1}{k-\ell-2s-1}\frac{(k-\ell-s-1)!}{s!}\, w^{\ell-k+2s+1}.\]
Using again the first step, this expression is readily seen to be
equal to $(-1)^{k+\ell+1}\mathfrak{J}(\ell,k)\newline=\mathfrak{J}(-k,-\ell)$.
\end{proof}

\section{The characteristic function for the antisymmetric diagonal}

For a given $d\in\mathbb{Z}_{+}$ let $E_{\pm}$ denote the $(2d+1)\times(2d+1)$
matrix with units on the upper (lower) parallel to the diagonal and
with all other entries equal to zero. Hence\[
(E_{+})_{j,k}=\delta_{j+1,k},\mbox{ }(E_{-})_{j,k}=\delta_{j,k+1},\mbox{ }j,k=-d,-d+1,-d+2,\ldots,d.\]
For $y=(y_{-d},y_{-d+1},y_{-d+2},\dots,y_{d})\in\mathbb{C}^{2d+1}$
let $\diag(y)$ denote the diagonal $(2d+1)\times(2d+1)$ matrix with
the sequence $y$ on the diagonal. Everywhere in what follows, $I$
stands for a unit matrix.

First a formula is presented for the determinant of a Jacobi matrix
with a general diagonal but with constant neighboring parallels to
the diagonal. As explained in the subsequent remark, however, this
formula can be extended to the general case with the aid of a simple
decomposition of the Jacobi matrix in question.
\begin{prop}
For $d\in\mathbb{N}$, $w\in\mathbb{C}$ and $y=(y_{-d},y_{-d+1},y_{-d+2},\ldots,y_{d})\in\mathbb{C}^{2d+1}$,
$\prod_{k=1}^{d}y_{k}y_{-k}\neq0$, one has\begin{eqnarray}
 &  & \hs{-1.5em}\det\!\big(\diag(y)+wE_{+}+wE_{-}\big)=\left(\prod_{k=1}^{d}y_{k}y_{-k}\right)\!\bigg[y_{0}\,\mathfrak{F}\!\left(\frac{w}{y_{1}},\dots,\frac{w}{y_{d}}\right)\mathfrak{F}\!\left(\frac{w}{y_{-1}},\dots,\frac{w}{y_{-d}}\right)\nonumber \\
 &  & \hs{0.3em}-\,\frac{w^{2}}{y_{1}}\,\mathfrak{F\!}\left(\frac{w}{y_{2}},\dots,\frac{w}{y_{d}}\right)\mathfrak{F}\!\left(\frac{w}{y_{-1}},\dots,\frac{w}{y_{-d}}\right)-\frac{w^{2}}{y_{-1}}\,\mathfrak{F}\!\left(\frac{w}{y_{1}},\dots,\frac{w}{y_{d}}\right)\mathfrak{F}\!\left(\frac{w}{y_{-2}},\dots,\frac{w}{y_{-d}}\right)\bigg]\!.\nonumber \\
 &  & \mbox{}\label{eq:det_Jacobid_gen}\end{eqnarray}
\end{prop}
\begin{proof}
Let us proceed by induction in $d$. The case $d=1$ is easy to verify.
Put $\mathcal{N}_{d}(w;y)=\det\!\big(\diag(y)+wE_{+}+wE_{-}\big)$.
Suppose (\ref{eq:det_Jacobid_gen}) is true for some $d\geq1$. For
given $w\in\mathbb{C}$ and $y\in\mathbb{C}{}^{2d+3}$ consider the
quantity $\mathcal{N}_{d+1}(w;y)$. Let us split the corresponding
$(2d+3)\times(2d+3)$ Jacobi matrix into four blocks by splitting
the set of indices into two disjoint sets $\{-d-1,d+1\}$ and $\{-d,-d+1,-d+2,\ldots,d\}$.
Applying the rule\[
\det\begin{pmatrix}A & B\\
C & D\end{pmatrix}=\det(A)\det(D-CA^{-1}B)\]
one derives the recurrence relation\[
\mathcal{N}_{d+1}(w;y_{-d-1},y_{-d},y_{-d+1},\ldots,y_{d},y_{d+1})=y_{d+1}y_{-d-1}\mathcal{N}_{d}(w;y'_{-d},y_{-d+1},y_{-d+2},\ldots,y_{d-1},y'_{d})\]
where\[
y'_{d}=\left(1-\frac{w^{2}}{y_{d}y_{d+1}}\right)\! y_{d},\mbox{ }y'_{-d}=\left(1-\frac{w^{2}}{y_{-d}y_{-d-1}}\right)\! y_{-d}.\]
Now it is sufficient to use the induction hypothesis jointly with
the equality\[
(1-x_{n-1}x_{n})\,\mathfrak{F}\!\left(x_{1},x_{2},\ldots,x_{n-2},\frac{x_{n-1}}{1-x_{n-1}x_{n}}\right)=\,\mathfrak{F}(x_{1},x_{2},\ldots,x_{n-1},x_{n})\]
which is valid for $n\geq2$ and which follows from relations (\ref{eq:Fx_reversed})
and (\ref{eq:Fx_eq_Fxk_Fxrem}), with $k=2$.\end{proof}
\begin{rem}
Let us consider a general finite symmetric Jacobi matrix $J$ of the
form\[
J=\begin{pmatrix}\lambda_{1} & w_{1}\\
w_{1} & \lambda_{2} & w_{2}\\
 & \ddots & \ddots & \ddots\\
 &  & \ddots & \ddots & \ddots\\
 &  &  & w_{n-2} & \lambda_{n-1} & w_{n-1}\\
 &  &  &  & w_{n-1} & \lambda_{n}\end{pmatrix}\]
such that $\prod_{k=1}^{n-1}w_{k}\neq0$. The Jacobi matrix can be
decomposed into the product\begin{equation}
J=G\tilde{J}G\label{eq:J_decompose}\end{equation}
where $G=\diag(\gamma_{1},\gamma_{2},\ldots,\gamma_{n})$ is a diagonal
matrix and $\tilde{J}$ is a Jacobi matrix with all units on the neighboring
parallels to the diagonal,\[
\tilde{J}=\begin{pmatrix}\tilde{\lambda}_{1} & 1\\
1 & \tilde{\lambda}_{2} & 1\\
 & \ddots & \ddots & \ddots\\
 &  & \ddots & \ddots & \ddots\\
 &  &  & 1 & \tilde{\lambda}_{n-1} & 1\\
 &  &  &  & 1 & \tilde{\lambda}_{n}\end{pmatrix}.\]
Hence $\det(J)=\left(\prod_{k=1}^{n}\gamma_{k}^{\,2}\right)\det(\tilde{J})$,
and one can employ formula (\ref{eq:det_Jacobid_gen}) to evaluate
$\det(\tilde{J})$ (in the case of odd dimension). In more detail,
one can put\[
\gamma_{2k-1}=\prod_{j=1}^{k-1}\frac{w_{2j}}{w_{2j-1}}\,,\mbox{ }\gamma_{2k}=w_{1}\prod_{j=1}^{k-1}\frac{w_{2j+1}}{w_{2j}}\,,\mbox{ }k=1,2,3,\ldots.\]
Alternatively, the sequence $\{\gamma_{k}\}_{k=1}^{n}$ is defined
recursively by $\gamma_{1}=1$, $\gamma_{k+1}=w_{k}/\gamma_{k}$.
Furthermore, $\tilde{\lambda}_{k}=\lambda_{k}/\gamma_{k}^{\,2}$.
With this choice, (\ref{eq:J_decompose}) is clearly true.
\end{rem}
Next we aim to derive a formula for the characteristic function of
a Jacobi matrix with an antisymmetric diagonal. Suppose $\lambda=(\lambda_{-d},\lambda_{-d+1},\lambda_{-d+2},\ldots,\lambda_{d})\in\mathbb{C}^{2d+1}$
and $\lambda_{-k}=-\lambda_{k}$ for $-d\leq k\leq d$; in particular,
$\lambda_{0}=0$. We consider the Jacobi matrix $K=\diag(\lambda)+wE_{+}+wE_{-}$.
Let us denote, temporarily, by $S$ the diagonal matrix with alternating
signs on the diagonal, $S=\diag(1,-1,1,\ldots,1)$, and by $Q$ the
permutation matrix with the entries $Q_{j,k}=\delta_{j+k,0}$ for
$-d\leq j,k\leq d$. The commutation relations\[
SQKQS=-K,\mbox{ }S^{2}=Q^{2}=I,\]
imply\[
\det(K-zI)=\det\!\big(SQ(K-zI)QS\big)=-\det(K+zI).\]
Hence the characteristic function of $K$ is an odd polynomial in
the variable $z$. This can be also seen from the explicit formula
(\ref{eq:char_antisym_diag}) derived below.
\begin{prop}
\label{thm:char_antisym_diag} Suppose $d\in\mathbb{N}$, $w\in\mathbb{C}$,
$\lambda\in\mathbb{C}^{2d+1}$ and $\lambda_{-k}=-\lambda_{k}$ for
$k=-d,-d+1,-d+2,\ldots,d$. Then\begin{eqnarray}
 &  & \hs{-1em}\frac{(-1)^{d+1}}{z}\,\det\!\big(\diag(\lambda)+wE_{+}+wE_{-}-zI\big)\label{eq:char_antisym_diag}\\
 &  & \hs{-1em}=\left(\prod_{k=1}^{d}(\lambda_{k}^{2}-z^{2})\right)\mathfrak{F}\!\left(\frac{w}{\lambda_{1}-z},\dots,\frac{w}{\lambda_{d}-z}\right)\mathfrak{F}\!\left(\frac{w}{\lambda_{1}+z},\dots,\frac{w}{\lambda_{d}+z}\right)\nonumber \\
 &  & \hs{-1em}\phantom{=}\,+2\sum_{j=1}^{d}w^{2j}\!\left(\prod_{k=j+1}^{d}(\lambda_{k}^{2}-z^{2})\right)\!\mathfrak{F}\!\left(\frac{w}{\lambda_{j+1}-z},\dots,\frac{w}{\lambda_{d}-z}\right)\mathfrak{F}\!\left(\frac{w}{\lambda_{j+1}+z},\dots,\frac{w}{\lambda_{d}+z}\right)\!.\nonumber \end{eqnarray}
\end{prop}
\begin{proof}
This is a particular case of (\ref{eq:det_Jacobid_gen}) where one
has to set $y_{k}=\lambda_{k}-z$ for $k>0$, $y_{0}=-z$, $y_{k}=-(\lambda_{-k}+z)$
for $k<0$. To complete the proof it suffices to verify that\begin{eqnarray*}
 &  & \frac{w^{2}}{z\,(\lambda_{1}-z)}\,\mathfrak{F}\!\left(\frac{w}{\lambda_{2}-z},\dots,\frac{w}{\lambda_{d}-z}\right)\mathfrak{F}\!\left(\frac{w}{\lambda_{1}+z},\frac{w}{\lambda_{2}+z},\dots,\frac{w}{\lambda_{d}+z}\right)\\
 &  & -\frac{w^{2}}{z\,(\lambda_{1}+z)}\,\mathfrak{F}\!\left(\frac{w}{\lambda_{1}-z},\frac{w}{\lambda_{2}-z},\dots,\frac{w}{\lambda_{d}-z}\right)\mathfrak{F}\!\left(\frac{w}{\lambda_{2}+z},\dots,\frac{w}{\lambda_{d}+z}\right)\\
 &  & =\,2\sum_{j=1}^{d}w^{2j}\!\left(\prod_{k=1}^{j}\frac{1}{\lambda_{k}^{2}-z^{2}}\right)\!\mathfrak{F}\!\left(\frac{w}{\lambda_{j+1}-z},\dots,\frac{w}{\lambda_{d}-z}\right)\mathfrak{F}\!\left(\frac{w}{\lambda_{j+1}+z},\dots,\frac{w}{\lambda_{d}+z}\right)\!.\end{eqnarray*}
To this end, one can apply (\ref{eq:uFuFv_antisym}), with $n=d$,
$u_{k}=w/(\lambda_{k}-z)$, $v_{k}=w/(\lambda_{k}+z)$. Note that
$u_{j}-v_{j}=2zu_{j}v_{j}/w$. 
\end{proof}
Zero always belongs to spectrum of the Jacobi matrix $K$ for the
characteristic function is odd. Moreover, as is well known and as
it simply follows from the analysis of the eigenvalue equation, if
$w\neq0$ then to every eigenvalue of $K$ there belongs exactly one
linearly independent eigenvector.
\begin{prop}
Suppose $w\in\mathbb{C}$, $\lambda\in\mathbb{C}^{2d+1}$, $\lambda_{-k}=-\lambda_{k}$
for $-d\leq k\leq d$, and $\prod_{k=1}^{d}\lambda_{k}\neq0$. Then
the vector $v\in\mathbb{C}^{2d+1}$, $v^{T}=(\theta_{-d},\theta_{-d+1},\theta_{-d+2},\ldots,\theta_{d})$,
with the entries\begin{equation}
\theta_{k}=(-1)^{k}w^{k}\!\left(\prod_{j=k+1}^{d}\lambda_{j}\right)\mathfrak{F}\!\left(\frac{w}{\lambda_{k+1}},\frac{w}{\lambda_{k+2}},\ldots,\frac{w}{\lambda_{d}}\right)\mbox{ for }k=0,1,2,\ldots,d,\label{eq:0eigenv_antisym_diag}\end{equation}
$\theta_{-k}=(-1)^{k}\,\theta_{k}$ \hspace{0.3em} for $-d\leq k\leq d$,
\hspace{0.5em} belongs to the kernel of the Jacobi matrix $\diag(\lambda)+wE_{+}+wE_{-}$.
In particular,\hspace{0.6em} $\theta_{0}=\lambda_{1}\lambda_{2}\ldots\lambda_{d}\,\mathfrak{F}(w/\lambda_{1},w/\lambda_{2},\ldots,w/\lambda_{d})$,
$\theta_{d}=(-1)^{d}w^{d}$, and so $v\neq0$.\end{prop}
\begin{rem*}
Clearly, formulas (\ref{eq:0eigenv_antisym_diag}) can be extended
to the case $\prod_{k=1}^{d}\lambda_{k}=0$ as well provided one makes
the obvious cancellations.\end{rem*}
\begin{proof}
One has to show that\[
w\theta_{k-1}+\lambda_{k}\theta_{k}+w\theta_{k+1}=0,\mbox{ }k=-d+1,-d+2,\ldots,d-1,\]
and $\lambda_{-d}\theta_{-d}+w\theta_{-d+1}=0$, $w\theta_{d-1}+\lambda_{d}\theta_{d}=0$.
Owing to the symmetries $\lambda_{-k}=-\lambda_{k}$, $\theta_{-k}=(-1)^{k}\theta_{k}$,
it suffices to verify the equalities only for indices $0\leq k\leq d$.
This can be readily carried out using the explicit formulas (\ref{eq:0eigenv_antisym_diag})
and the rule (\ref{eq:F_finite_recur}).
\end{proof}

\section{Jacobi matrices with a linear diagonal}

Finally we focus on finite-dimensional Jacobi matrices of odd dimension
whose diagonal depends linearly on the index and whose parallels to
the diagonal are constant. Without loss of generality one can assume
that the diagonal equals\newline $(-d,-d+1,-d+2,\ldots,d)$, $d\in\mathbb{Z}_{+}$.
For $w\in\mathbb{C}$ put\[
K_{0}=\diag(-d,-d+1,-d+2,\ldots,d),\mbox{ }K(w)=K_{0}+wE_{+}+wE_{-}.\]
Concerning the characteristic function $\chi(z)=\det(K(w)-z)$, we
know that this is an odd function. Put\[
\chi_{\mathrm{red}}(z)=\frac{(-1)^{d+1}}{z}\,\det(K(w)-z).\]
Hence $\chi_{\mathrm{red}}(z)$ is an even polynomial of degree $2d$.
Further, denote by\linebreak $\{e_{-d},e_{-d+1},e_{-d+2},\ldots,e_{d}\}$
the standard basis in $\mathbb{C}^{2d+1}$.

Suppose $w\neq0$. Let us consider a family of column vectors $x_{s,n}\in\mathbb{C}^{2d+1}$
depending on the parameters $s,n\in\mathbb{Z}$ and defined by\[
x_{s,n}^{\,\, T}=\big(\mathfrak{J}(s+d,n),\mathfrak{J}(s+d-1,n),\mathfrak{J}(s+d-2,n),\ldots,\mathfrak{J}(s-d,n)\big).\]
From the fact that the matrix $\mathfrak{J}$ obeys (\ref{eq:A_order2_diff_eq}),
(\ref{eq:A_antisymm}) one derives that\[
\forall s,n\in\mathbb{Z},\mbox{ }K(w)x_{s,n}=s\, x_{s,n}-w\,\mathfrak{J}(s+d+1,n)e_{-d}-w\,\mathfrak{J}(s-d-1,n)e_{d}.\]
Put\[
v_{s}=x_{s,s+d+1},\mbox{ }s\in\mathbb{Z}.\]
Recalling that $\mathfrak{J}(m,m)=\mathfrak{J}(-m,m)=0$ one has\begin{equation}
K(w)v_{s}=s\, v_{s}-w\,\mathfrak{J}(s-d-1,s+d+1)e_{d}.\label{eq:Kw_vs}\end{equation}

\begin{rem}
Putting $s=0$ one gets $K(w)v_{0}=0$, and so $v_{0}$ spans the
kernel of $K(w)$.\end{rem}
\begin{lem}
\label{thm:basis_Vcal} For every $\ell=-d,-d+1,-d+2,\ldots,d,$ one
has\[
w^{d+\ell}\sum_{s=-d}^{\ell}\,\frac{(-1)^{\ell+s}}{(d+s)!\,(\ell-s)!}\, v_{s}\in e_{\ell}+\span\{e_{\ell+1},e_{\ell+2},\ldots,e_{d}\}.\]
In particular,\begin{equation}
e_{d}=w^{2d}\sum_{s=-d}^{d}\,\frac{(-1)^{d+s}}{(d+s)!\,(d-s)!}\, v_{s}.\label{eq:ed_basis_Vcal}\end{equation}
Consequently, $\mathcal{V}=\{v_{-d},v_{-d+1},v_{-d+2},\ldots,v_{d}\}$
is a basis in $\mathbb{C}^{2d+1}$.\end{lem}
\begin{proof}
One has to show that\[
w^{d+\ell}\sum_{s=-d}^{\ell}\,\frac{(-1)^{\ell+s}}{(d+s)!\,(\ell-s)!}\,\mathfrak{J}(s-k,s+d+1)=\delta_{\ell,k}\mbox{ }\mbox{ for }-d\leq k\leq\ell.\]
Note that for any $a\in\mathbb{C}$ and $n\in\mathbb{Z}_{+}$,\[
\sum_{k=0}^{n}(-1)^{k}\binom{n}{k}\binom{a+k}{r}=0,\mbox{ }r=0,1,2,\ldots,n-1,\mbox{ }\sum_{k=0}^{n}(-1)^{k}\binom{n}{k}\binom{a+k}{n}=(-1)^{n}.\]
Using these equalities and (\ref{eq:Jmn}) one can readily show, more
generally, that\[
\sum_{s=-d}^{\ell}\,\frac{(-1)^{\ell+s}}{(d+s)!\,(\ell-s)!}\,\mathfrak{J}(m+s,n+s)=0\mbox{ }\mbox{ for }m,n\in\mathbb{Z},m\leq n\leq m+d+\ell,\]
and\[
\sum_{s=-d}^{\ell}\,\frac{(-1)^{\ell+s}}{(d+s)!\,(\ell-s)!}\,\mathfrak{J}(m+s,m+d+\ell+s+1)=w^{-d-\ell}.\]
This proves the lemma.
\end{proof}
Denote by $\tilde{K}(w)$ the matrix of $K(w)$ in the basis $\mathcal{V}$
introduced in Lemma~\ref{thm:basis_Vcal}. Let $a,b\in\mathbb{C}^{2d+1}$
be the column vectors defined by $a^{T}=(\alpha_{-d},\alpha_{-d+1},\alpha_{-d+2},\ldots\alpha_{d})$,
$b^{T}=(\beta_{-d},\beta_{-d+1},\beta_{-d+2},\ldots\beta_{d})$,\begin{equation}
\alpha_{s}=\mathfrak{J}(s-d-1,s+d+1),\mbox{ }\beta_{s}=\frac{(-1)^{d+s}\, w^{2d+1}}{(d+s)!\,(d-s)!}\,,\mbox{ }s=-d,-d+1,-d+2,\ldots,d.\label{eq:alpha_beta_def}\end{equation}
Note that\begin{equation}
\alpha_{-s}=-\alpha_{s},\mbox{ }\beta_{-s}=\beta_{s}.\label{eq:alpha_beta_sym}\end{equation}
The former equality follows from (\ref{eq:A_odd_each_var}) and (\ref{eq:A_antisymm}).
From (\ref{eq:Kw_vs}) and (\ref{eq:ed_basis_Vcal}) one deduces that\begin{equation}
\tilde{K}(w)=K_{0}-ba^{T}.\label{eq:Ktilde_K0_rank1}\end{equation}
Note, however, that the components of the vectors $a$ and $b$ depend
on $w$, too, though not indicated in the notation.

According to (\ref{eq:Ktilde_K0_rank1}), $\tilde{K}(w)$ differs
from the diagonal matrix $K_{0}$ by a rank-one correction. This form
is suitable for various computations. Particularly, one can express
the resolvent of $\tilde{K}(w)$ explicitly,\[
(\tilde{K}(w)-z)^{-1}=(K_{0}-z)^{-1}+\frac{1}{1-a^{T}(K_{0}-z)^{-1}b}\,(K_{0}-z)^{-1}ba^{T}(K_{0}-z)^{-1}.\]
The equality holds for any $z\in\mathbb{C}$ such that $z\notin\spec\{K_{0}\}=\{-d,-d+1,-d+2,\ldots,d\}$
and $1-a^{T}(K_{0}-z)^{-1}b\neq0$. Clearly, this set of excluded
values of $z$ is finite.

Let us proceed to derivation of a formula for the characteristic function
of $K(w)$. Proposition~\ref{thm:char_antisym_diag} is applicable
to $K(w)$ and so\begin{eqnarray}
 &  & \hs{-1.7em}\chi_{\mathrm{red}}(z)\,=\left(\prod_{k=1}^{d}(k^{2}-z^{2})\right)\mathfrak{F}\!\left(\frac{w}{1-z},\dots,\frac{w}{d-z}\right)\mathfrak{F}\!\left(\frac{w}{1+z},\dots,\frac{w}{d+z}\right)\label{eq:char_linear_diag}\\
 &  & +\,2\sum_{j=1}^{d}w^{2j}\!\left(\prod_{k=j+1}^{d}(k^{2}-z^{2})\right)\!\mathfrak{F}\!\left(\frac{w}{j+1-z},\dots,\frac{w}{d-z}\right)\mathfrak{F}\!\left(\frac{w}{j+1+z},\dots,\frac{w}{d+z}\right)\!.\nonumber \end{eqnarray}
Below we derive a more convenient formula for $\chi_{\mathrm{red}}(z)$.
\begin{lem}
One has\begin{equation}
\chi_{\mathrm{red}}(0)=\sum_{s=0}^{d}\frac{\big((d-s)!\big)^{2}\,(2d-s+1)!}{s!\,(2d-2s+1)!}\, w^{2s}\label{eq:char_red_0}\end{equation}
and\begin{equation}
\chi_{\mathrm{red}}(n)=\frac{1}{n}\,\sum_{k=0}^{n-1}\,(-1)^{k}(2k+1)!\binom{n+k}{2k+1}\binom{d+k+1}{2k+1}w^{2d-2k}\label{eq:char_red_n}\end{equation}
for $n=1,2,\ldots,d$.\end{lem}
\begin{proof}
Let us first verify the formula for $\chi_{\mathrm{red}}(0)$. From
(\ref{eq:char_linear_diag}) it follows that\[
\chi_{\mathrm{red}}(0)=(d!)^{2}\,\mathfrak{F}\!\left(w,\frac{w}{2},\dots,\frac{w}{d}\right)^{\!2}+2\sum_{j=1}^{d}w^{2j}\!\left(\frac{d!}{j!}\right)^{\!2}\mathfrak{F}\!\left(\frac{w}{j+1},\frac{w}{j+2},\dots,\frac{w}{d}\right)^{\!2}.\]
By Proposition~\ref{thm:BesselJY_rel_F},\[
\chi_{\mathrm{red}}(0)=\pi^{2}w^{2d+2}\, Y_{d+1}(2w)^{2}\!\left(J_{0}(2w)^{2}+2\sum_{j=1}^{d}J_{j}(2w)^{2}\right)\!+O\!\left(w^{2d+2}\log(w)\right).\]
Further we need some basic facts concerning Bessel functions; see,
for instance, \cite[Chp.~9]{AbramowitzStegun}. Recall that\[
J_{0}(z)^{2}+2\sum_{j=1}^{\infty}J_{j}(z)^{2}=1.\]
Hence\begin{eqnarray*}
\chi_{\mathrm{red}}(0) & = & \pi^{2}w^{2d+2}\, Y_{d+1}(2w)^{2}+O\!\left(w^{2d+2}\log(w)\right)\\
 & = & \!\left(\sum_{k=0}^{d}\,\frac{(d-k)!}{k!}\, w^{2k}\right)^{\!2}+O\!\left(w^{2d+2}\log(w)\right).\end{eqnarray*}
Note that $\chi_{\mathrm{red}}(0)$ is a polynomial in the variable
$w$ of degree $2d$, and so\[
\chi_{\mathrm{red}}(0)=\sum_{s=0}^{d}\sum_{k=0}^{s}\,\frac{(d-k)!\,(d-s+k)!}{k!\,(s-k)!}\, w^{2s}.\]
Using the identity\begin{eqnarray*}
\sum_{k=0}^{s}\,\frac{(d-k)!\,(d-s+k)!}{k!\,(s-k)!} & = & \big((d-s)!\big)^{2}\,\sum_{k=0}^{s}\binom{d-k}{d-s}\binom{d-s+k}{d-s}\\
 & = & \frac{\big((d-s)!\big)^{2}\,(2d-s+1)!}{s!\,(2d-2s+1)!}\end{eqnarray*}
one arrives at (\ref{eq:char_red_0}).

To show (\ref{eq:char_red_n}) one can make use of (\ref{eq:Ktilde_K0_rank1}).
One has\[
\chi_{\mathrm{red}}(z)=\frac{(-1)^{d+1}}{z}\,\det(\tilde{K}(w)-z)=\frac{(-1)^{d+1}}{z}\,\det(K_{0}-z)\det\!\left(I-(K_{0}-z)^{-1}ba^{T}\right).\]
Note that $\det(I+ba^{T})=1+a^{T}b$. Hence, in view of (\ref{eq:alpha_beta_sym}),\[
\chi_{\mathrm{red}}(z)=\prod_{k=1}^{d}(k^{2}-z^{2})\left(1-\sum_{s=-d}^{d}\,\frac{\beta_{s}\alpha_{s}}{s-z}\right)=\prod_{k=1}^{d}(k^{2}-z^{2})\left(1-2\,\sum_{s=1}^{d}\,\frac{s\beta_{s}\alpha_{s}}{s^{2}-z^{2}}\right).\]
Using (\ref{eq:alpha_beta_def}) one gets\[
\chi_{\mathrm{red}}(n)=-2n\beta_{n}\alpha_{n}\prod_{\begin{array}{c}
{k=1\atop k\neq n}\end{array}}^{d}(k^{2}-n^{2})=\frac{(-1)^{d}}{n}\, w^{2d+1}\,\mathfrak{J}(n-d-1,n+d+1).\]
Formula (\ref{eq:char_red_n}) then follows from (\ref{eq:Jmn}).\end{proof}
\begin{prop}
For every $d\in\mathbb{Z}_{+}$ one has\begin{equation}
\chi_{\mathrm{red}}(z)=\sum_{s=0}^{d}\binom{2d-s+1}{s}w^{2s}\,\prod_{k=1}^{d-s}(k^{2}-z^{2}).\label{eq:char_final}\end{equation}
\end{prop}
\begin{proof}
Since $\chi_{\mathrm{red}}(z)$ is an even polynomial in $z$ of degree
$2d$ it is enough to check that the RHS of (\ref{eq:char_final})
coincides, for $z=0,1,2,\dots,d$, with $\chi_{\mathrm{red}}(0)$,
$\chi_{\mathrm{red}}(1)$, $\chi_{\mathrm{red}}(2)$, $\dots$, $\chi_{\mathrm{red}}(d)$.
With the knowledge of values (\ref{eq:char_red_0}) and (\ref{eq:char_red_n}),
this is a matter of straightforward computation.\end{proof}
\begin{rem}
Using (\ref{eq:char_final}) it is not difficult to check that formula
(\ref{eq:char_red_n}) is valid for any $n\in\mathbb{N}$, including
$n>d$ (the summation index $k$ runs from 1 to $\min\{n-1,d\}$).
\end{rem}
\begin{rem}
If $w\in\mathbb{R}$, $w\neq0$, then the spectrum of the Jacobi matrix
$K(w)$ is real and simple, and formula (\ref{eq:char_final}) implies
that the interval $[-1,1]$ contains no other eigenvalue except of
$0$.
\end{rem}
Eigenvectors of $K(w)$ can be expressed in terms of the function
$\mathfrak{F}$, too. Suppose $w\neq0$. Let us introduce the vector-valued
function $x(z)\in\mathbb{C}^{2d+1}$ depending on $z\in\mathbb{C}$,
$x(z)^{T}=(\xi_{-d}(z),\xi_{-d+1}(z),\xi_{-d+2}(z),\ldots,\xi_{d}(z))$,\[
\xi_{k}(z)=w^{-d-k}\,\frac{\Gamma(z+d+1)}{\Gamma(z-k+1)}\,\mathfrak{F}\!\left(\frac{w}{z-k+1},\frac{w}{z-k+2},\ldots,\frac{w}{z+d}\right)\!,\mbox{ }-d\leq k\leq d.\]
With the aid of (\ref{eq:F_finite_recur}) one derives the equality\begin{equation}
\big(K(w)-z\big)x(z)=-w^{-2d}\,\frac{\Gamma(z+d+1)}{\Gamma(z-d)}\,\mathfrak{F}\!\left(\frac{w}{z-d},\frac{w}{z-d+1},\ldots,\frac{w}{z+d}\right)e_{d}.\label{eq:Kw_xz}\end{equation}

\begin{rem}
According to (\ref{eq:Fw_over_nu_trunc}),\[
\xi_{k}(z)=w^{-d-k}\,\sum_{s=0}^{[(d+k)/2]}\,(-1)^{s}\,\frac{(d+k-s)!}{s!\,(d+k-2s)!}\, w^{2s}\,\prod_{j=s}^{d+k-s-1}(z+d-j).\]
Hence $\xi_{k}(z)$ is a polynomial in $z$ of degree $d+k$. In particular,
$\xi_{-d}(z)=1$, and so $x(z)\neq0$.\end{rem}
\begin{prop}
One has\begin{equation}
\chi(z)=-z\left(\prod_{k=1}^{d}(z^{2}-k^{2})\right)\mathfrak{F}\!\left(\frac{w}{z-d},\frac{w}{z-d+1},\frac{w}{z-d+2},\ldots,\frac{w}{z+d}\right).\label{eq:char_rel_F}\end{equation}
If $w\in\mathbb{C}$, $w\neq0$, then for every eigenvalue $\lambda\in\spec(K(w))$,
$x(\lambda)$ is an eigenvector corresponding to $\lambda$.\end{prop}
\begin{proof}
Denote by $P(z)$ the RHS of (\ref{eq:char_rel_F}). By (\ref{eq:Kw_xz}),
if $P(\lambda)=0$ then $x(\lambda)$ is an eigenvector of $K(w)$.
Thus it suffices to verify (\ref{eq:char_rel_F}). The both sides
depend on $w$ polynomially and so it is enough to prove the equality
for $w\in\mathbb{R}\setminus\{0\}$. Note that $P(z)$ is a polynomial
in $z$ of degree $2d+1$, and the coefficient standing at $z^{2d+1}$
equals $-1$. The set of roots of $P(z)$ is contained in $\spec(K(w))$.
One can show that $P(z)$ has no multiple roots. In fact, suppose
$P(\lambda)=P'(\lambda)=0$ for some $\lambda\in\mathbb{R}$. From
(\ref{eq:Kw_xz}) one deduces that $(K(w)-\lambda)x(\lambda)=0$,
$(K(w)-\lambda)x'(\lambda)=x(\lambda)$ (here $x'(z)$ is the derivative
of $x(z)$). Hence\[
(K(w)-\lambda)^{2}x(\lambda)=(K(w)-\lambda)^{2}x'(\lambda)=0.\]
Note that $x'(\lambda)\neq0$, and $x'(\lambda)$ differs from a multiple
of $x(\lambda)$ for $\xi_{-d}(z)=1$. This contradicts the fact,
however, that the spectrum of $K(w)$ is simple. One concludes that
the set of roots of $P(z)$ coincides with $\spec(K(w))$. Necessarily,
$P(z)$ is equal to the characteristic function of $K(w)$.\end{proof}
\begin{rem}
With the aid of (\ref{eq:char_rel_F}) one can rederive equality (\ref{eq:char_red_n}).
For $1\leq n\leq d$, a straightforward computation gives\[
\chi(n)=(-1)^{d+n}w^{2d+1}\,\mathfrak{J}(d-n+1,d+n+1).\]
Equality (\ref{eq:char_red_n}) then follows from (\ref{eq:Jmn}).
\end{rem}

\section*{Acknowledgments}

The authors wish to acknowledge gratefully partial support from Grant
No. 201/09/0811 of the Czech Science Foundation (P.\v{S}.) and from
Grant No. LC06002 of the Ministry of Education of the Czech Republic
(F.\v{S}.).

\end{document}